\documentclass[a4paper,aps,twocolumn,pre]{revtex4-1}
\usepackage{graphicx}
\usepackage{amsmath}           % essential
\usepackage{amsfonts}          % AMS fonts
\usepackage{epstopdf}
\usepackage{amsthm}

\usepackage{newtxtext,newtxmath}

\usepackage{tikz,pgfplots,bbm,bm,mathrsfs,dsfont}
\usetikzlibrary{shapes,snakes,arrows,pgfplots.groupplots,fadings,calc,decorations.markings,positioning,chains,fit,shapes,calc}

\tikzset{
    >=triangle 45,
    box/.style={draw = black, rectangle, rounded corners, inner sep=4pt,fill=white,text=black, text centered},
    arro/.style={line width=1pt,gray,->,shorten <=2pt,shorten >=1pt}
}

\usepackage{color}
\usepackage{quoting}

\usepackage[english]{babel}
\usepackage{hyperref}
\hypersetup{pdfauthor={Caracciolo DiGioacchino Gherardi Malatesta},pdftitle={Solution for a Bipartite Euclidean TSP in one dimension},%
	colorlinks, linktocpage=true, pdfstartpage=3, pdfstartview=FitV,%
	urlcolor=orange, linkcolor=blue, citecolor=red, %pagecolor=RoyalBlue,%
}

\newcommand{\idop}{\mathds{1}}
\newcommand{\be}{\begin{equation}}
\newcommand{\ee}{\end{equation}}
\def\reff#1{(\protect\ref{#1})}

\newtheorem{lemma}{Lemma}

\newtheorem{pros}{Proposition}[section]

\begin{document}
\title{Solution for a bipartite Euclidean traveling-salesman problem in one dimension}
\author{Sergio Caracciolo}\email{sergio.caracciolo@mi.infn.it}
\affiliation{Dipartimento di Fisica, University of Milan and INFN, via Celoria 16, 20133 Milan, Italy}
\author{Andrea Di Gioacchino}\email{andrea.digioacchino@unimi.it}
\affiliation{Dipartimento di Fisica, University of Milan and INFN, via Celoria 16, 20133 Milan, Italy}
\author{Marco Gherardi}\email{marco.gherardi@mi.infn.it}
\affiliation{Dipartimento di Fisica, University of Milan and INFN, via Celoria 16, 20133 Milan, Italy}
\author{Enrico M. Malatesta}\email{enrico.m.malatesta@gmail.com}
\affiliation{Dipartimento di Fisica, University of Milan and INFN, via Celoria 16, 20133 Milan, Italy}
	
%\section{}
%\subsection{}
%\tableofcontents

%\vfill\eject

%\begin{figure}[h!]
%\centering
%\includegraphics[width=0.9\columnwidth]{Foto.jpg}
%\end{figure}

\date{\today}
\begin{abstract}
The traveling salesman problem is one of the most studied combinatorial optimization problems, because of the simplicity in its statement and the difficulty in its solution.
We characterize the optimal cycle for every convex and increasing cost function when the points are thrown independently and with an identical probability distribution in a compact interval. We compute the average optimal cost for every number of points when the distance function is the square of the Euclidean distance.  We also show that the average optimal cost is not a self-averaging quantity by explicitly computing the variance of its distribution in the thermodynamic limit. Moreover, we prove that the cost of the optimal cycle is not smaller than twice the cost of the optimal assignment of the same set of points. Interestingly, this bound is saturated in the thermodynamic limit. 
\end{abstract}
\maketitle

\section{Introduction}

Given $N$ cities and $N (N-1)/2$ values that represent the cost paid for traveling between all pairs of them, the traveling salesman problem (TSP) consists in finding the tour that visits all the cities and finally comes back to the starting point with the least total cost to be paid for the journey.
The TSP is the archetypal problem in combinatorial optimization~\cite{lawler1985}. Its first formalization can be probably traced back to the Austrian mathematician Karl Menger, in the 1930s~\cite{menger1932}, but it is yet extensively investigated. As it belongs to the class of NP-complete problems, see Karp and Steele in~\cite{lawler1985}, the study of the TSP could shed light on the famous P vs NP problem~\footnote{http://www.claymath.org/millennium-problems/p-vs-np-problem}.
%, which turned out to be such a difficult problem that the Clay Mathematics Institute of Cambridge has inserted it among the millennium problems, offering a prize of \mbox{1 000 000} dollars for solving it. 
Many problems in various fields of science (computer science, operational research, genetics, engineering, electronics and so on) and in everyday life (lacing shoes, Google maps queries, food deliveries and so on) can be mapped on a TSP or a variation of it, see for example Ref.~\cite[Chap.~3]{reinelt1994}  for a non-exhaustive list.
Interestingly, the complexity of the TSP seems to remain high even if we try to modify the problem. For example, the Euclidean TSP, where the costs to travel from cities are the Euclidean distances between them, remains NP-complete~\cite{papadimitriou1977}. The bipartite TSP, where the cities are divided in two sub-sets and the tour has to alternate between them, is NP-complete too, as its Euclidean counterpart. 
%On the opposite, if we try to remove the ``tour'' request we obtain the matching problem (in its monopartite and bipartite version, or variants of them) which is in the P complexity class.
It is well known that the statistical properties of the optimal solution of problems in combinatorial optimization can be related to the zero temperature behaviour of corresponding disordered statistical mechanics models~\cite{Kirkpatrick1983,Sourlas1986,mezard1987spin,mezard2009information} when a class of problems is defined and a probability distribution for the different instances is precised. 
%This observation motivated the interest of the statistical physics community in this kind of problems. 
%A quite large class of these optimization problems is defined on a more or less complicated graph (the definition of the TSP as a graph problem is presented in Sec. \ref{sec:cost}), so the field of combinatorial optimization problems is of particular interest also from the network and complex systems point of view.

Previous investigations of some of us suggested that the Euclidean matching problem is simpler to deal with in its bipartite version. This idea encouraged us to consider the bipartite TSP, starting from the one dimensional case that is fully analyzed here.

The  manuscript is organized as follows: in Sect.~\ref{sec:cost} we define the  TSP and its variants we are interested in.  We shall introduce a representation of the model, which is novel as far as we know, in terms of a couple of permutations.
%, with the cost function which has to be minimized to find the solution of a given instance of it. We also introduce the mean cost, which is the main object of our work.
%In Sect.~\ref{sec:model} we develop our formalism and notation.
In this way we also establish a very general connection between the bipartite TSP and a much simpler model, which is in the P complexity class, the assignment problem.
Always using our representation, in Sect.~\ref{sec:sol} we can provide the explicit solution of the problem for every instance of the disorder (that is, for every position of the points) in the one dimensional case when the cost is a convex and increasing function of the Euclidean distance between the cities. 
In Sect.~\ref{sec:costev} we exploit our explicit solution to compute the average optimal cost for an arbitrary number of points, when they are chosen with uniform distribution in the unit interval, and we present a comparison with the results of numerical simulations.
%We also return on the correspondence between the bipartite TSP and the assignment problem, computing explicitly the costs of the assignments which compose the optimal solution.
In Sect.~\ref{sec:asymptotic} we discuss the behaviour of the cost in the thermodynamic limit of an infinite number of points. 
Here the results can be extended to more general distribution laws for the points.
In Sect.~\ref{sec:conclusions} we give our conclusions.

\section{The Model}\label{sec:cost}
%\subsection{The cost function}

Given a generic (undirected) graph $\mathcal{G} = (\mathcal{V}, \mathcal{E})$, a {\em cycle}  of length $k$ is a sequence of edges $e_1, e_2, \dots, e_k \in \mathcal{E}$ in which two subsequent edges $e_i$ and $e_{i+1}$ share a vertex for $i=1, \dots, k$ where, for $i=k$ the edge $e_{k+1}$ must be identified with the edge $e_1$. On a bipartite graph each cycle must have an even length. The cycle is {\em Hamiltonian} when the visited vertices are all different and the cardinality of the set of vertices $|\mathcal{V}|$ is exactly $k$ for $k>2$. In other terms, a Hamiltonian cycle is a closed path visiting all the vertices in $\mathcal{V}$ only once. The determination of the existence of an Hamiltonian cycle is an NP-complete problem (see Johnson and Papadimitriou in~\cite{lawler1985}).
A graph that contains a Hamiltonian cycle is called a Hamiltonian graph.
The complete graph with $N$ vertices $\mathcal{K}_N$ is Hamiltonian for $N>2$.
The bipartite complete graph with $N+M$ vertices $\mathcal{K}_{N,M}$ is Hamiltonian for $M=N>1$.

Let us denote by $\mathcal H$ the set of Hamiltonian cycles of the graph $\mathcal{G}$. Let us suppose now that a weight $w_e > 0$ is assigned to each edge $e \in \mathcal{E}$ of the graph $\mathcal{G}$. We can associate to each Hamiltonian cycle $h\in \mathcal{H}$ a total cost
\be
E(h) :=  \sum_{e\in h} w_e \, .\label{E}
\ee
In the (weighted) Hamiltonian cycle problem we search for the Hamiltonian cycle $h\in \mathcal{H}$ such that the total cost in~\reff{E} is minimized, i.e., the optimal Hamiltonian cycle $h^*\in \mathcal{H}$ is such that
\be
E(h^*) = \min_{ h\in \mathcal{H}} E(h)\, . \label{h^*}
\ee
When the $N$ vertices of $\mathcal{K}_N$ are seen as cities and the weight for each edge is the cost paid to cover the route distance between the cities, the search for $h^*$ is called the {\em traveling salesman problem} (TSP). For example, consider when the graph $\mathcal{K}_N$ is embedded in $\mathbb{R}^d$, that is for each $i\in [N] =\{1,2,\dots,N\}$ we associate a point $x_i\in \mathbb{R}^d$,  and for $e=(i,j)$ with $i,j \in [N]$ we introduce a cost which is a function of their Euclidean distance $w_e = |x_i-x_j|^p$ with $p\in \mathbb{R}$. When $p=1$, we obtain the usual \emph{Euclidean} TSP.
Analogously for the bipartite graph $\mathcal{K}_{N,N}$ we will have two sets of points in $\mathbb{R}^d$, that is the red $\{r_i\}_{i\in [N]}$ and the blue $\{b_i\}_{i\in [N]}$ points and the edges connect red with blue points with a cost
\be
w_e = |r_i-b_j|^p \, . \label{pb}
\ee
When $p=1$, we obtain the usual \emph{bipartite} Euclidean TSP.
The simplest way to introduce randomness in the problem is to consider the weights $w_e$ independent and identically distributed random variables. In this case the problem is called {\em random} TSP and has been extensively studied by disordered system techniques such as replica and cavity methods ~\cite{Vannimenus1984,Orland1985,Sourlas1986,Mezard1986,Mezard1986a,Krauth1989,Ravanbakhsh2014} and by a rigorous approach \cite{Wastlund}. In the random Euclidean TSP~\cite{BHH,Steele,Karp,Percus1996,Cerf1997}, instead, the positions of the points are generated at random and as a consequence the weights will be correlated.
The typical properties of the optimal solution are of interest, and in particular the average optimal cost
\be
\overline{E}  := \overline{E(h^*)}\,,
\ee
where we have denoted by a bar the average over all possible realization of the disorder.

\subsection{Representation in terms of permutations}\label{sec:model}

We shall now restrict to the complete bipartite graph $\mathcal{K}_{N,N}$. Let $\mathcal{S}_N$ be the group of permutation of $N$ elements. For each $\sigma, \pi\in  \mathcal{S}_N$, the sequence for $i\in [N]$
\be
\begin{aligned}
e_{2i-1} = & \, (r_{\sigma(i)}, b_{\pi(i)} )  \\
e_{2i} = & \, (b_{\pi(i)}, r_{\sigma(i+1)}) \label{corris}
\end{aligned}
\ee
where $\sigma(N+1)$ must be identified with $\sigma(1)$, defines a Hamiltonian cycle.
More properly, it defines a Hamiltonian cycle with starting vertex $r_1=r_{\sigma(1)}$ with a particular orientation, that is 
\be
h[(\sigma, \pi)] := (r_1 b_{\pi(1)} r_{\sigma(2)} b_{\pi(2)} \cdots r_{\sigma(N)} b_{\pi(N)}) 
%= (r_{\sigma(1)} C )
= (r_{1} C ) \,,
\ee
where $C$ is an open walk which visit once all the blue points and all the red points with the exception of $r_1$. Let $C^{-1}$ be the open walk in opposite direction. This defines a new, dual, couple of permutations which generate the same Hamiltonian cycle
\be
h[(\sigma, \pi)^\star] := (C^{-1} r_1) = (r_1 C^{-1} ) = h[(\sigma, \pi)]\, ,
\ee
since the cycle $(r_1 C^{-1} )$ is the same as $(r_1 C)$ (traveled in the opposite direction).
By definition
\begin{align}
\begin{split}
&h[(\sigma, \pi)^\star] \\
& = (r_1 b_{\pi(N)} r_{\sigma(N)} b_{\pi(N-1)} r_{\sigma(N-1)} \cdots b_{\pi(2)}  r_{\sigma(2)} b_{\pi(1)} ) \, .
\end{split}
\end{align}
Let us introduce the cyclic permutation $\tau \in \mathcal{S}_N$, which performs a left rotation, and the inversion $I \in \mathcal{S}_N$. That is $\tau(i) = i+1$ for $i\in [N-1]$ with  $\tau(N) = 1$ and $I(i) = N+1 -i$. In the following we shall denote a permutation by using the second raw in the usual two-raw notation, that is, for example $\tau = (2, 3, \cdots ,N, 1)$ and $I= (N, N-1, \dots, 1)$. Then
\be
h[(\sigma, \pi)^\star] =  h[( \sigma \circ \tau \circ I, \pi \circ I)] \, . \label{d1}
\ee
There are $N! \, (N-1)!/2$ Hamiltonian cycles for $\mathcal{K}_{N,N}$.
Indeed the couples of permutations are $(N!)^2$ but we have to divide them by $2N$ because of the $N$ different starting points and the two directions in which the cycle can be traveled. 

%Indeed there are $(N!)^2$ couples of permutations; $N$ of them give rise to the same ordered cycle (correspond to different starting points). However, half of them are one the dual of the other and correspond to the two different directions in which the cycle can be formed.
%There are $2\, N!$ couples of permutations, $2\, N$ of them give rise to the same ordered cycle bur correspond to different starting points, 2 of them are one the dual of the other and correspond to the two different directions in which the cycle can be formed.

\subsection{Comparison with the assignment problem}

From~\reff{corris} and weights of the form~\reff{pb}, we get an expression for the total cost
\begin{align}\label{costgen}
\begin{split}
& E[h[(\sigma, \pi)]] \\
%& = \sum_{i \in [N]}\left[  |r_{\sigma(i)} - b_{\pi(i)}|^p + |r_{\sigma(i+1)} - b_{\pi(i)}|^p \right] \\
&=  \sum_{i \in [N]}\left[  |r_{\sigma(i)} - b_{\pi(i)}|^p + |r_{\sigma \circ \tau(i)} - b_{\pi(i)}|^p \right] \, .
\end{split}
\end{align} 
Now we can re-shuffle the sums and we get
\begin{align}
\begin{split}
& E[h[(\sigma, \pi)]] \\
& = \sum_{i \in [N]}  |r_{i} - b_{\pi \circ \sigma^{-1}(i)}|^p + 
 \sum_{i \in [N]}  |r_{i} - b_{\pi \circ \tau^{-1} \circ \sigma^{-1}(i)}|^p \\
& = E[m(\pi \circ \sigma^{-1})] + E[m(\pi \circ \tau^{-1} \circ \sigma^{-1})] \label{dec}
\end{split}
\end{align}
where $E[m(\lambda)]$ is the total cost of the assignment $m$ in $\mathcal{K}_{N,N}$ associated to the permutation $\lambda\in \mathcal{S}_N$
\be
E[m(\lambda)] = \sum_{i \in [N]}  |r_{i} - b_{\lambda(i)}|^p \, .
\ee
The duality transformation~\reff{d1}, that is
\begin{align}
\sigma \, \to & \;  \sigma \circ \tau \circ I \\
\pi \, \to & \; \pi \circ I \,,
\end{align}
interchanges the two matchings because
\begin{subequations}
\begin{align}
\begin{split}
\mu_1 := \pi \circ \sigma^{-1} \, \to & \;  \pi \circ I \circ I \circ \tau^{-1} \circ \sigma^{-1} \\
&= \pi \circ \tau^{-1} \circ \sigma^{-1}
\end{split}
\\
\begin{split}
\mu_2 := \pi \circ \tau^{-1} \circ \sigma^{-1}  \, \to & \; \pi \circ I \circ \tau^{-1} \circ I \circ \tau^{-1} \circ \sigma^{-1} \\
&%= \pi \circ \tau \circ \tau^{-1} \circ \sigma^{-1} 
= \pi \circ \sigma^{-1} 
\end{split}
\end{align}
\end{subequations}
where we used 
\be
 I \circ \tau^{-1} \circ I = \tau \, .\label{Itau}
\ee
The two matchings corresponding to the two permutations $\mu_1$ and $\mu_2$ have no edges in common and therefore each vertex will appear twice in the union of their edges. Remark also that
\be
\mu_2 = \mu_1 \circ \sigma \circ \tau^{-1} \circ \sigma^{-1}
\ee
which means that $\mu_1$ and $\mu_2$ are related by a permutation which has to be, as it is $\tau^{-1}$, a unique cycle of length $N$. It follows that, if $h^*$ is the optimal Hamiltonian cycle and $m^*$ is the optimal assignment,
\be
E[h^*] \; \geq \; 2 \, E[m^*] \, . \label{stima}
\ee
In the case of the Euclidean assignment the scaling of the average optimal cost is known in every dimensions and for every $p>1$ \cite{Caracciolo:158}:
%It is relevant to remember what we know on the asymptotic behaviour of the average total cost for the optimal Euclidean assignment in all dimensions $d$
\be
\overline{E[\mu^*]} \sim
\begin{cases}
N^{1-\frac{p}{2}} & d=1\, ; \\
N^{1-\frac{p}{2}} (\log N)^{\frac{p}{2}} & d=2\, ; \\
N^{1-\frac{p}{d}} & d>2\,.
\end{cases}
\ee
The scaling shows an anomalous behaviour at lower dimension differently from what occurs for the matching problem on the complete graph $\mathcal{K}_N$ where in any dimension the scaling with the number of points is always $N^{1-\frac{p}{d}}$. Indeed, also for the monopartite Euclidean TSP (that is  on $\mathcal{K}_N$) in~\cite{BHH} it has been shown that for $p=1$, in a finite region, with probability 1, the total cost scales according to $N^{1-\frac{p}{d}}$ in any dimension.

\section{Solution in $d=1$ for all instances}\label{sec:sol}

Here we shall concentrate on the one-dimensional case, where both red and blue points are chosen uniformly in the unit interval $[0,1]$.
In our analysis we shall make use of the results for the Euclidean assignment problem in one dimension of~\cite{Caracciolo:159} which have been obtained when in~\reff{pb} is set $p>1$. In this work it is showed that sorting both red and blue points in increasing order, the optimal assignment is defined by the identity permutation $\idop=(1,2,\dots, N)$. 
%It is convenient to sort both red and blue points in increasing order. Then the optimal assignment is defined by the identity permutation $\idop=(12\dots N)$.
From now on, we will assume $p>1$ and that both red and blue points are ordered, i.e. $r_1 \le \dots \le r_N$ and $b_1 \le \dots \le b_N$.
Let 
\be
\tilde{\sigma}(i) = 
\begin{cases}
2i-1 & i \leq  (N+1)/2 \\
2N -2i +2 & i > (N+1)/2 \label{sigmatilde}
\end{cases}
\ee
and%, if we introduce the permutation $I=(N\dots 21)$, 
\begin{equation}
\tilde{\pi}(i) = \tilde{\sigma}\circ I(i) = \tilde{\sigma}(N+1-i)  =
\begin{cases}
2i  & i < (N +1)/2 \\
2N - 2i +1 & i \geq  (N +1)/2 \label{pitilde}
\end{cases}
\end{equation}
the couple $(\tilde{\sigma}, \tilde{\pi})$  will define a Hamiltonian cycle $\tilde{h}\in \mathcal{H}$. More precisely, according to the correspondence given in~\reff{corris}, it contains the edges
for even $N$, 
\begin{subequations}
\begin{align}
\tilde{e}_{2i-1} = & \,  
\begin{cases}
(r_{2i-1}, b_{2i})   & i \leq  N/2 \\
(r_{2N-2i+2}, b_{2N-2i +1})    & i > N/2
\end{cases} \\
\tilde{e}_{2i} = & \,  \begin{cases}
(b_{2i}, r_{2i+1})   & i <  N/2 \\
(b_{N}, r_{N}) & i = N/2 \\
(b_{2N-2i+1}, r_{2N-2i})    &  N/2< i < N \\
(b_{1}, r_{1})    &  i = N
\end{cases}
\end{align}
\end{subequations}
while for  $N$ odd
\begin{subequations}
\begin{align}
\tilde{e}_{2i-1} = & \,  
\begin{cases}
(r_{2i-1}, b_{2i})   & i <  (N-1)/2 \\
(r_{N}, b_{N}) & i = (N-1)/2 \\
(r_{2N-2i+2}, b_{2N-2i +1})    & i > (N-1)/2
\end{cases} \\
\tilde{e}_{2i} = & \,  \begin{cases}
(b_{2i}, r_{2i+1})   & i <  (N-1)/2 \\
(b_{2N-2i+1}, r_{2N-2i})    &  (N-1)/2< i < N \\
(b_{1}, r_{1})    &  i = N \, .
\end{cases}
\end{align}
\end{subequations}

\begin{figure}
	\centering
	\includegraphics[width=0.5\columnwidth]{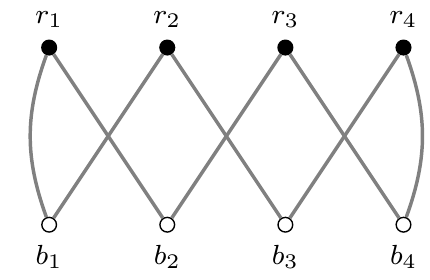}
\caption{The optimal Hamiltonian cycle $\tilde{h}$ for $N=4$ blue and red points chosen in the unit interval and sorted in increasing order.} \label{exh}
\end{figure}

The main ingredient of our analysis is the following
\begin{pros}
For a convex and increasing cost function the optimal Hamiltonian cycle is provided by $\tilde{h}$.
\end{pros}

This cycle is the analogous of the \emph{criss-cross} solution introduced by~Halton~\cite{halton1995} (see Fig.~\ref{exh}). 
In his work, Halton studied the optimal way to lace a shoe. This problem can be seen as a peculiar instance of a 2-dimensional bipartite Euclidean TSP with the parameter which tunes the cost $p=1$.
One year later, Misiurewicz~\cite{misiurewicz1996} generalized Halton's result giving the least restrictive requests on the 2-dimensional TSP instance to have the criss-cross cycle as solution.
Other generalizations of these works have been investigated in more recent papers~\cite{polster2002,garcia2017}. We will show that the same criss-cross cycle has the lowest cost for the Euclidean bipartite TSP in one dimension, provided that $p>1$. 
To do this, we will prove in a novel way the optimality of the criss-cross solution, suggesting two moves that lower the energy of a tour and showing that the only Hamiltonian cycle that cannot be modified by these moves is $\tilde{h}$.

We shall make use of the following moves in the ensemble of Hamiltonian cycles.
Given $i, j\in [N]$ with $j>i$ we can partition each cycle as
\be
h[(\sigma, \pi)] = (C_1 r_{\sigma(i)} b_{\pi(i)} C_2 b_{\pi(j)} r_{\sigma(j+1)} C_3),
\ee
where the $C_i$ are open paths in the cycle, and we can define the operator $R_{ij}$ that exchanges two blue points $b_{\pi(i)}$ and $b_{\pi(j)}$ and reverses the path between them as
\begin{align}
\begin{split}
h[R_{ij} (\sigma, \pi)] &:=   (C_1 r_{\sigma(i)} [b_{\pi(i)} C_2 b_{\pi(j)}]^{-1} r_{\sigma(j+1)} C_3) \\
& = (C_1 r_{\sigma(i)} b_{\pi(j)} C_2^{-1} b_{\pi(i)} r_{\sigma(j+1)} C_3)  \,.
\end{split}
\end{align}
Analogously by writing
\be
h[(\sigma, \pi)] = (C_1 b_{\pi(i-1)} r_{\sigma(i)}  C_2  r_{\sigma(j)} b_{\pi(j)} C_3)
\ee
we can define the corresponding operator $S_{ij}$ that exchanges two red points $r_{\sigma(i)}$ and $r_{\sigma(j)}$ and reverses the path between them
\begin{align}
\begin{split}
h[S_{ij} (\sigma, \pi)] &:=   (C_1 b_{\pi(i-1)} [r_{\sigma(i)} C_2 r_{\sigma(j)}]^{-1} b_{\pi(j)} C_3) \\
& = (C_1  b_{\pi(i-1)} r_{\sigma(j)} C_2^{-1}  r_{\sigma(i)} b_{\pi(j)} C_3) \, .
\end{split}
\end{align}
Two couples of points $(r_{\sigma(k)}, r_{\sigma(l)})$ and $(b_{\pi(j)}, b_{\pi(i)})$ have the same orientation 
if $(r_{\sigma(k)} -r_{\sigma(l)})(b_{\pi(j)}-b_{\pi(i)}) >0$. Remark that as we have ordered both set of points this means also that $(\sigma(k), \sigma(l))$ and $(\pi(j), \pi(i))$ have the same orientation.

Then
\begin{lemma}\label{lemma1}
Let $E[(\sigma, \pi)]$ be the cost defined in~\reff{costgen}. Then 
$E[R_{ij}(\sigma, \pi)] - E[(\sigma, \pi)] >0$ if the couples $(r_{\sigma(j+1)}, r_{\sigma(i)})$ and $(b_{\pi(j)}, b_{\pi(i)})$ have the same orientation and
 $E[S_{ij}(\sigma, \pi)] - E[(\sigma, \pi)] >0$ if the couples $(r_{\sigma(j)}, r_{\sigma(i)})$ and $(b_{\pi(j)}, b_{\pi(i-1)})$ have the same orientation.
\end{lemma}
\begin{proof}
\begin{multline}
\begin{split}
E[R_{ij}(\sigma, \pi)] - E[(\sigma, \pi)] &= w_{(r_{\sigma(i)}, b_{\pi(j)})}  + w_{(b_{\pi(i)}, r_{\sigma(j+1)})} \\
&- w_{(r_{\sigma(i)}, b_{\pi(i)}) }- w_{(b_{\pi(j)}, r_{\sigma(j+1)})}
\end{split}
\end{multline}
and this is the difference between two matchings which is positive if the couples $(r_{\sigma(j+1)}, r_{\sigma(i)})$ and $(b_{\pi(j)}, b_{\pi(i)})$ have the same orientation (as shown in~\cite{McCannRobert1999, Caracciolo:159} for a weight which is an increasing convex function of the Euclidean distance). 
%These are $N$ conditions.
%The indices of $\sigma$ and $\pi$ have to be interpreted modulus $N$. 
The remaining part of the proof is analogous.

\end{proof}

\begin{lemma}
The only couples of permutations $(\sigma,\pi)$ with $\sigma(1)=1$ such that both $(\sigma(j+1), \sigma(i))$ have the same orientation  as $(\pi(j), \pi(i))$  and $(\pi(j), \pi(i-1))$ and $(\sigma(j), \sigma(i))$, for each $i, j\in [N]$ are $(\tilde{\sigma},\tilde{\pi})$ and its dual $(\tilde{\sigma},\tilde{\pi})^\star$.
\end{lemma}
\begin{proof}
We have to start  our Hamiltonian cycle from $r_{\sigma(1)} = r_1$.  
Next we look at $\pi(N)$,  if we assume now that $\pi(N)>1$, there will be  a $j$ such that our cycle would have the form $(r_1 C_1 r_{\sigma(j)} b_1 C_2   b_{\pi(N)})$, if we assume $j>1$ then $(1, \sigma(j))$ and $(\pi(N),1)$ have opposite orientation, so that  necessarily $\pi(N)=1$. In the case $j=1$  our Hamiltonian cycle is of the form $(r_1 b_1 C)$, that is $(b_1 C r_1)$, and this is exactly of the other form if we exchange red and blue points.
%Therefore we can assume, without loss of generality that it is of the form $(r_1 C b_1)$.
We assume that it is of the form $(r_1 C b_1)$; the other form would give, at the end of the proof, $(\tilde{\sigma},\tilde{\pi})^\star$.
\\
Now we shall proceed by induction. Assume that our Hamiltonian cycle is of the form $(r_1 b_2 r_3 \cdots x_k C y_k \cdots b_3 r_2 b_1)$ with $k<N$, where $x_k$ and $y_k$ are, respectively, a red point and a blue point when $k$ is odd and  viceversa when $k$ is even. 
Then $y_{k+1}$ and $x_{k+1}$ must be in the walk $C$. 
If $y_{k+1}$ it is not the point on the right of $x_k$ the cycle has the form 
$(r_1 b_2 r_3 \cdots x_k y_s C_1 y_{k+1} x_l \cdots y_k \cdots b_3 r_2 b_1)$ 
but then $(x_l , x_k)$ and $(y_{k+1}, y_s)$ have opposite orientation, 
which is impossible, so that 
$s=k+1$, that is the point on the right of $x_k$. Where is $x_{k+1}$? If it is not the point on the left of $y_k$ the cycle has the form $(r_1 b_2 r_3 \cdots x_k y_{k+1} \cdots y_l x_{k+1} C_1 x_s \cdots y_k \cdots b_3 r_2 b_1)$, but then $(x_s, x_{k+1})$ and $(y_k, y_l)$ have opposite orientation, which is impossible, so that $s =k+1$, that is the point on the left of $y_k$. We have now shown that the cycle has the form  $(r_1 b_2 r_3 \cdots y_{k+1} C x_{k+1} \cdots b_3 r_2 b_1)$ and can proceed until $C$ is empty.
\end{proof} 
The case with $N=3$ points is explicitly investigated in appendix~\ref{app:n3}.

Now that we have understood what is the optimal Hamiltonian cycle, we can look in more details at what are the two matchings which enter in the decomposition we used in~\reff{dec}.
As $\tilde{\pi} = \tilde{\sigma} \circ I$ we have that
\be
I = \tilde{\sigma}^{-1} \circ \tilde{\pi} = \tilde{\pi}^{-1} \circ \tilde{\sigma}.
\ee
As a consequence both permutations associated to the matchings appearing in~\reff{dec} for the optimal Hamiltonian cycle are involutions:
\begin{subequations}
\begin{align}
\begin{split}
\tilde{\mu}_1 \equiv \tilde{\pi} \circ \tilde{\sigma}^{-1} & = \tilde{\sigma} \circ I \circ \tilde{\sigma}^{-1} = \tilde{\sigma} \circ  \tilde{\pi}^{-1} \\
& = \left[ \tilde{\pi} \circ \tilde{\sigma}^{-1}\right]^{-1}\label{m1}
\end{split}
\\
\begin{split}
\tilde{\mu}_2 \equiv \tilde{\pi} \circ \tau^{-1} \circ \tilde{\sigma}^{-1} & = \tilde{\sigma} \circ I \circ \tau^{-1} \circ I \circ \tilde{\pi}^{-1} \\ 
& = \left[\tilde{\pi} \circ \tau^{-1} \circ \tilde{\sigma}^{-1} \right]^{-1}  \label{m2},
\end{split}
\end{align}
\end{subequations}
where we used~\reff{Itau}.
This implies that those two permutations have at most cycles of period two, a fact which reflects a symmetry by exchange of red and blue points.

When $N$ is odd it happens that
\be
I \circ \tilde{\sigma} \circ I = \tilde{\sigma} \circ \tau^{-\frac{N-1}{2}},
\ee
so that
\begin{align}
\begin{split}
I \circ \tilde{\pi} \circ I & = I \circ \tilde{\sigma} \circ I \circ I = \tilde{\sigma} \circ \tau^{-\frac{N-1}{2}} \circ I\\ 
& = \tilde{\pi} \circ I \circ \tau^{-\frac{N-1}{2}} \circ I = \tilde{\pi} \circ \tau^{\frac{N-1}{2}} \,.
\end{split}
\end{align}
It follows that the two permutations in~\reff{m1} and~\reff{m2} are conjugate by $I$
\be
I \circ \tilde{\pi} \circ \tau^{-1} \circ \tilde{\sigma}^{-1} \circ I  =  \tilde{\pi}\circ \tau^{\frac{N-1}{2}} \circ \tau \circ \tau^{\frac{N-1}{2}} \circ \tilde{\sigma}^{-1} = \tilde{\pi} \circ \tilde{\sigma}^{-1}
\ee
so that, in this case, they have exactly the same numbers of cycles of order 2.
Indeed we have
\begin{subequations}
\begin{align}
\tilde{\mu}_1 = %\tilde{\pi} \circ \tilde{\sigma}^{-1} =
& \, (2,1,4,3,6, \dots , N-1,N-2,N)\\
\tilde{\mu}_2 = %\tilde{\pi} \circ \tau^{-1} \circ \tilde{\sigma}^{-1} =
&  \, (1,3,2,5,4, \dots N, N-1)
\end{align} 
\end{subequations}
and they have $\frac{N-1}{2}$ cycles of order 2 and 1 fixed point. See Fig.~\ref{N5}  for the case $N=5$.
\begin{figure}
	\centering
	\includegraphics[width=0.6\columnwidth]{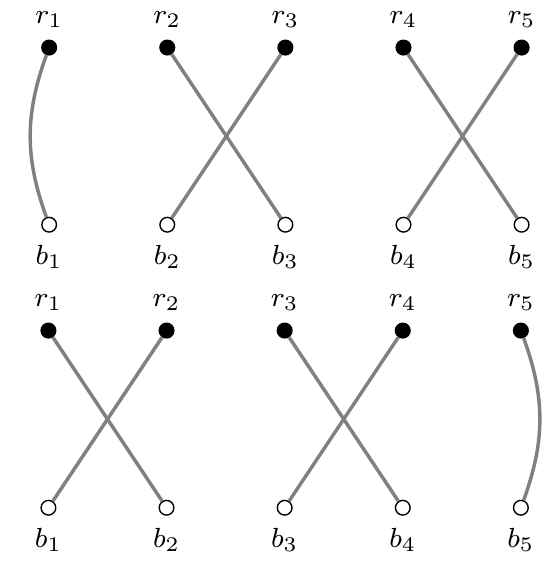}
\caption{Decomposition of the optimal Hamiltonian cycle $\tilde{h}$ for $N=5$ in two disjoint matchings $\tilde{\mu}_2$ and $\tilde{\mu}_1$.} \label{N5}
\end{figure}

In the case of even $N$ the two permutations have not the same number of cycles of order 2, indeed one has no fixed point and the other has two of them. More explicitly
\begin{subequations}
\begin{align}
\tilde{\mu}_1 = %\tilde{\pi} \circ \tilde{\sigma}^{-1}=
& \, (2,1,4,3,6, \dots ,N,N-1)\\
\tilde{\mu}_2 = %\tilde{\pi} \circ \tau^{-1} \circ \tilde{\sigma}^{-1}  = 
&  \, (1,3,2,5,4, \dots N-1,N-2, N)
\end{align} 
\end{subequations}
See Fig.~\ref{N4} for the case $N=4$.

\begin{figure}
	\centering
	\includegraphics[width=0.5\columnwidth]{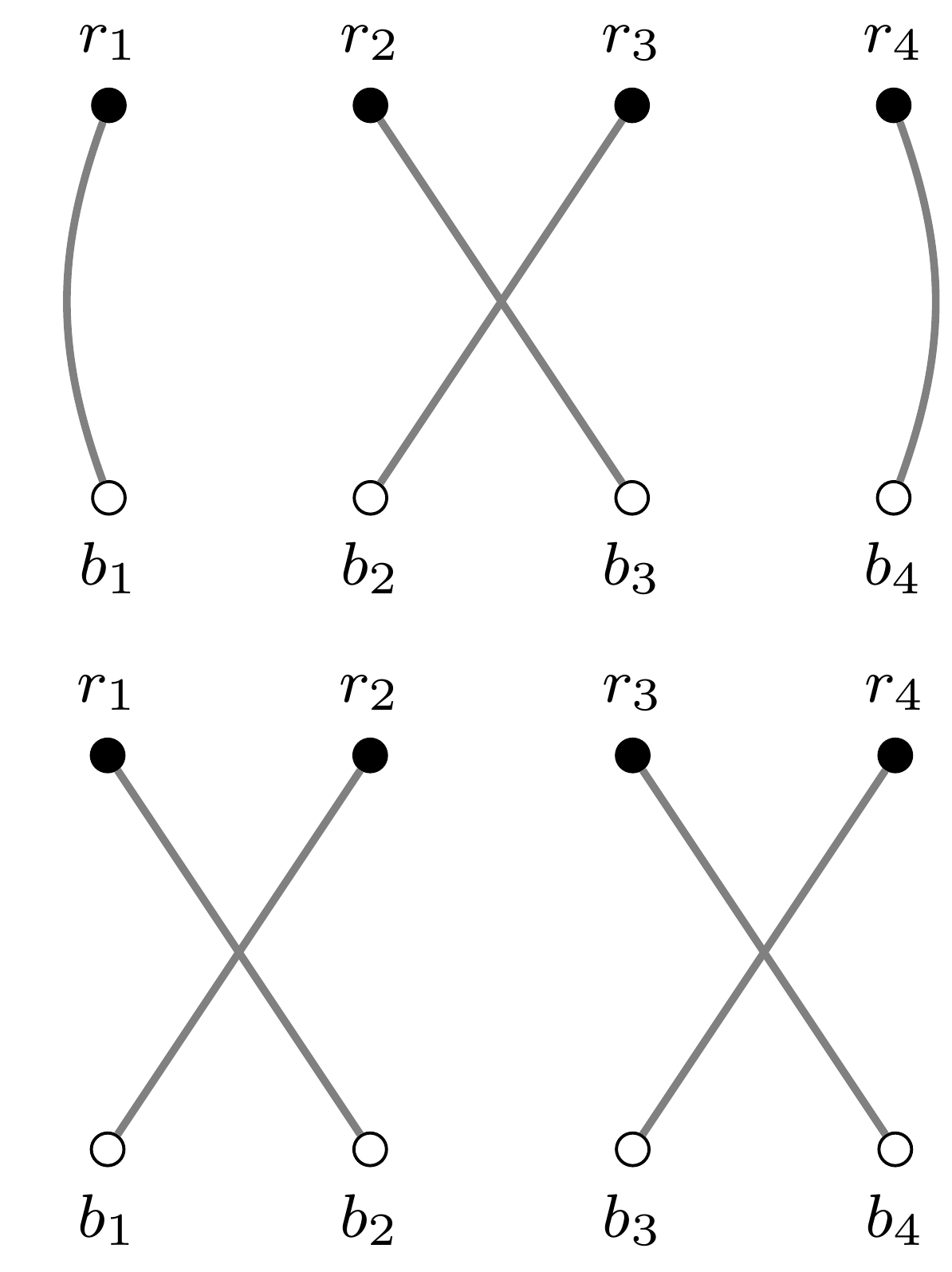}
\caption{Decomposition of the optimal Hamiltonian cycle $\tilde{h}$ for $N=4$ in he two disjoint matchings $\tilde{\mu}_2$ and $\tilde{\mu}_1$.}\label{N4}
\end{figure}

\section{Evaluation of the cost}\label{sec:costev}

Here we will evaluate the cost of the optimal Hamiltonian cycle $\tilde{h}$ for $\mathcal{K}_{N,N}$,
\begin{align}
\begin{split}
%E_N(h^*) =E_N(\tilde{h}) =   |r_1-b_1|^p + |r_N-b_N|^p +\sum_{i=1}^{N-1}\left[ |b_{i+1} - r_i|^p  + |r_{i+1} - b_i|^p \right] \label{EN}.
E_N(\tilde{h}) & = |r_1-b_1|^p + |r_N-b_N|^p \\
& +\sum_{i=1}^{N-1}\left[ |b_{i+1} - r_i|^p  + |r_{i+1} - b_i|^p \right] \label{EN}.
\end{split}
\end{align}
Assume that both red and blue points are chosen according to the law $\rho$ and let
\be
\Phi_\rho(x) := \int_0^x ds \, \rho(s) 
\ee
be its {\em cumulative}. The probability that, chosen $N$ points at random, the $k$-th is in the interval $(x, x + dx)$ is given by
\begin{align}
\Pr{}_\rho\left[x_k\in d x\right] = \, &  k\, \binom{N}{k} \Phi_\rho^{k-1}(x) \left[1- \Phi_\rho(x)\right]^{N-k}  \rho(x)\, d x \, .
\label{x_k}
\end{align}
In particular for $k=1$
\be
\Pr{}_\rho\left[x_1\in d x\right] = \,   N\,  \left[1-\Phi_\rho(x)\right]^{N-1}  \rho(x)\, d x \, .\label{lp1}
\ee
and $k=N$
\be
\Pr{}_\rho\left[x_N\in d x\right] = \,   N\,  \Phi_\rho^{N-1}(x)   \rho(x)\, d x \, .\label{lpd}
\ee
Given two sequences of $N$ points, the probability for the difference $\phi_k$ in the position  between the $(k+1)$-th and the $k$-th points is
\begin{align} 
\begin{split}
& \Pr{}_\rho\left[\phi_k\in d \phi\right] = k(k+1)\, \binom{N}{k} \, \binom{N}{k+1} \,d\phi_k \\
& \int dx \,dy  \,   \rho(x)\, \rho(y) \delta(\phi_k - y +x) \, \Phi_\rho(y) \, \left[1- \Phi_\rho(x)\right] \\
& \left[\Phi_\rho(x) \Phi_\rho(y) \right]^{k-1} \left[\left(1- \Phi_\rho(x)\right)\left( 1- \Phi_\rho(y)\right)\right]^{N-k-1}  \, . \label{phi_k}
\end{split}
\end{align}
Let us now focus on the simple case in which the law $\rho$ is flat, then $\Phi_\rho(x) = x$.
\be
\begin{aligned}
\overline{| r_1 - b_1|^p} & = N^2\, \int_0^1 dx\, dy\,  [(1-x)(1-y)]^{N-1} |x-y|^p \\
 & = N^2\,  \int_0^1 dx\, dy\, (x y)^{N-1} |x-y|^p \\
 & = \overline{| r_N - b_N|^p}.
\end{aligned}
\ee
For $p = 2$
\be 
\overline{| r_1 - b_1|^2} %N^2\,  \int_0^1 dx\, dy\, (x y )^{N-1} (x-y)^2 
= \frac{ 2 N }{ (N+1)^2 (N+2) }
\ee
and
\begin{align}
\begin{split}
& \overline{ |b_{k+1} - r_k|^2} = \overline{ |r_{k+1} - b_k|^2} = k(k+1)\, \binom{N}{k} \, \binom{N}{k+1}  \\
& \int_0^1 dx\, dy\, (x-y)^2 y (1-x) (xy)^{k-1}[(1-x)(1-y)]^{N-k-1}  \\
& = \frac{ 2 (k+1) (N-k+1)}{(N+1)^2 (N+2)} 
\end{split}
\end{align}
and
\be
\sum_{k=1}^{N-1} \frac{ 2 (k+1) (N-k+1)}{(N+1)^2 (N+2)}  = \frac{1}{3} \frac{(N+6)(N-1)}{(N+1)(N+2)}\, .
\ee
In conclusion, the average cost for the flat distribution and $p=2$ is exactly
\be
\overline{E_N^{(2)}} = \frac{2}{3} \frac{ N^2 + 4 N  - 3}{(N+1)^2} \, . \label{exact}
\ee
If we recall that for  the assignment the average optimal total cost is exactly $\frac{1}{3}\frac{N}{N+1}$, the difference between the average optimal total cost of the bipartite TSP and twice the assignment is
\be
 \frac{2}{3} \left[ \frac{ N^2 + 4 N  - 3}{(N+1)^2} - \frac{N}{N+1} \right] = \frac{1}{3} \frac{N-1}{(N+1)^2} \geq  0
 \ee
and vanishes for infinitely large $N$. Remark that the limiting value is reached from above for the TSP and from below for the assignment.
We plot in Fig.~\ref{plot} the numerical results of the average optimal cost for different number of points.

It is also interesting to look at the contribution from the two different matchings in which we have subdivided the optimal Hamiltonian cycle.
In the case of $N$ odd we have for one of them the average cost
\be
\frac{ 2 N }{ (N+1)^2 (N+2) }+  2 \sum_{k=1}^{\frac{N-1}{2} } \frac{ 4 k (N-2 k+2)}{(N+1)^2 (N+2)}  = \frac{1}{3} \frac{ N^2 + 4 N  - 3}{(N+1)^2}
\ee
and also for the other
\begin{align}
\begin{split}
&\frac{ 2 N }{ (N+1)^2 (N+2)} +  2 \sum_{k=1}^{\frac{N-1}{2}} \frac{ 2 (2 k+1) (N-2 k+1)}{(N+1)^2 (N+2)}  \\
& = \frac{1}{3} \frac{ N^2 + 4 N  - 3}{(N+1)^2} \, .
\end{split}
\end{align}
In the case of $N$ even we have for the matching with two fixed points the average cost
\begin{align}
\begin{split}
& \frac{ 4 N }{ (N+1)^2 (N+2) } +  2 \sum_{k=1}^{\frac{N-2}{2} } \frac{ 2 (2 k +1)  (N-2 k+1)}{(N+1)^2 (N+2)} \\
& = \frac{1}{3} \frac{ N^2 + 4 N  - 6}{(N+1)^2} \,,
\end{split}
\end{align}
while for the other with no fixed points
\be
2 \sum_{k=1}^{\frac{N-2}{2} } \frac{ 4 k   (N-2 k+2)}{(N+1)^2 (N+2)}  = \frac{1}{3} \frac{ N^2 + 4 N}{(N+1)^2},
\ee
which then has a cost higher at the order $N^{-2}$.

\begin{figure}
\centering
\includegraphics[width=0.95\columnwidth]{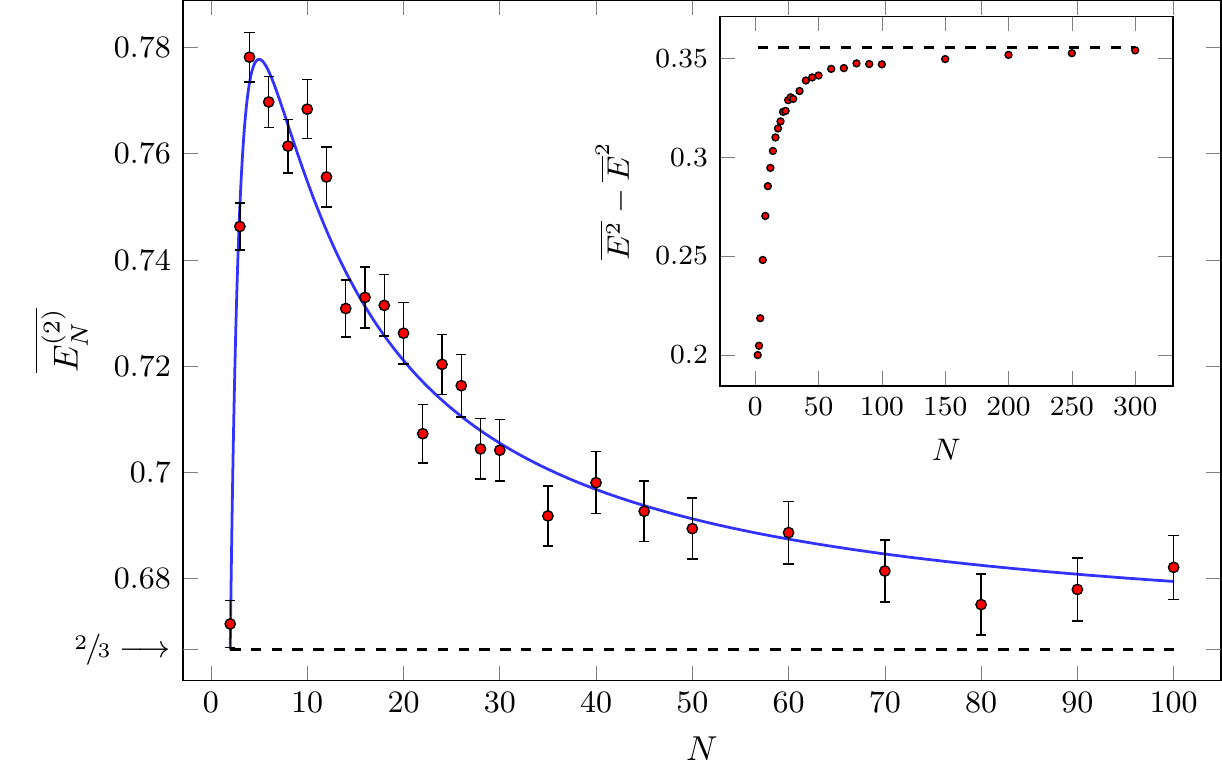}
\caption{Numerical results for $\overline{E_N^{(2)}}$ for several values of $N$. The continuous line represents the exact prediction given in~\reff{exact} and the dashed line gives the value for infinitely large $N$. For every $N$ we have used $10^4$ instances. In the inset we show the numerical results for the variance of the cost $E_N^{(2)}$ obtained using the exact solution provided by (\ref{sigmatilde}) and (\ref{pitilde}). The dashed line represents the theoretical large $N$ asymptotic value. Error bars are also plotted but they are smaller than the mark size.}\label{plot}
\end{figure}

%\section{Asymptotic analysis}\label{sec:asymptotic}
%
%In the limit of large $N$ only the term with a sum on $i$ in~\reff{EN} will contribute, and each of the two terms will provide the contribution of an optimal matching. Therefore, if we define
%\be
%\epsilon^{(p)}_N = \frac{1}{N} \overline{E_N(h^*)}
%\ee
%for $p\ge 1$ we can use the results from~\cite{Caracciolo:159, Caracciolo:160, Caracciolo:169} to get an expression for the asymptotic behaviour of the costs
%\begin{align}
%\lim_{N\to \infty} N^{\frac{p}{2}} \, \epsilon^{(p)}_N = \,  & 2\, \frac{2^p}{\sqrt{\pi}}\,  \Gamma\left(\frac{p+1}{2}\right) \, \int_0^1 ds\, \frac{\left[s (1-s)\right]^{\frac{p}{2}}}{\rho^p\left[\Phi^{-1}_\rho(s) \right]} \label{e1}\\
%= \,  & 2\, \frac{2^p}{\sqrt{\pi}}\,  \Gamma\left(\frac{p+1}{2}\right) \, \int_{-\infty}^\infty dt \frac{\left\{\Phi_\rho(t) \left[ 1 - \Phi_\rho(t) \right]\right\}^\frac{p}{2}} {\rho^{p-1}(t)} \label{e2}
%\end{align}
%where we performed the change of variables $s = \Phi_\rho(t)$. 

\section{Asymptotic analysis for the optimal average cost}\label{sec:asymptotic}

%The cost can be written as
%\begin{equation}
%E = \sum_{i=1}^{N} \, \left| m_{\mu_1}(x_) \right|^p
%\end{equation}

%In this section we will analyze what happens in the thermodynamic limit, using 
Motivated by the preceding discussion, one can try to perform a more refined analysis in the thermodynamic limit. In the asymptotic regime of large $N$, in fact, only the term with a sum on $i$ in~\reff{EN} will contribute, and each of the two terms will provide an equal optimal matching contribution. Proceeding as in the case of the assignment ~\cite{Caracciolo:159, Caracciolo:160}, one can show that the random variables $\phi_k$ defined above Eq.~(\ref{phi_k}) converge (in a weak sense specified by Donsker's theorem) to $\phi(s)$, which is a difference of two Brownian bridge processes~\cite{Caracciolo:169}.

One can write the re-scaled average optimal cost as
\begin{equation}
\overline{E_p} \equiv \lim_{N\to \infty} N^{\frac{p}{2}-1} \, \overline{ E_N^{(p)} } 
%=2 \int_0^1 \overline{ \left| \phi(s)\right|^p } \, \mathrm{d}s 
\label{Ep}
\end{equation}
where we have denoted with a bar $\overline{\; \cdot \;}$ the average over all the %realization of $\phi(s)$, i.e. over the probability density function
instances. 
By starting at finite $N$ with the representation~\reff{phi_k},
the large $N$ limit  can be obtained setting $k=Ns+\frac{1}{2}$ and introducing the variables $\xi$, $\eta$ and $\varphi$ such that
\begin{equation}
 x=s+\frac{\xi}{\sqrt N},\quad y=s+\frac{\eta}{\sqrt N},\quad \phi_k=\frac{\varphi(s)}{\sqrt N},
\end{equation}
in such a way that $s$ is kept fixed when $N\to +\infty$. Using the fact that
\begin{equation}
 \Phi_\rho^{-1}(x)\approx \Phi_\rho^{-1}\left(s+\frac{\xi}{\sqrt N}\right)=\Phi_\rho^{-1}(s)+\frac{\xi}{\sqrt N\left(\rho\circ\Phi_\rho^{-1}\right)(s)},
\end{equation}
we obtain, at the leading order,
\begin{equation}
\begin{split}
 \Pr & \left[\varphi(s)\in  d\varphi\right]= \\
=&\,d\varphi\iint\delta\left(\varphi-\frac{\eta-\xi}{\rho\left(\Phi_\rho^{-1}(s)\right)}\right)\frac{\exp\left(-\frac{\xi^2+\eta^2}{2 s(1-s)}\right)}{2\pi s(1-s)}d\xi \,d\eta\\
=&\,\frac{\left(\rho\circ\Phi_\rho^{-1}\right)(s)}{\sqrt{4 \pi s(1-s)}}\exp\left\{-\frac{\left[\left(\rho\circ\Phi_\rho^{-1}\right)(s)\right]^2}{4 s(1-s)} \varphi ^2\right\}d\varphi, \label{varphi-d}
\end{split}
\end{equation}
that implies that
\begin{equation}
\begin{split}
\overline{E_p} & =2 \int_0^1 \overline{ \left| \varphi(s)\right|^p } \, \mathrm{d}s \\
&=2 \, \int_0^1d s\frac{s^\frac{p}{2}(1-s)^\frac{p}{2}}{\left[\left(\rho\circ\Phi_\rho^{-1}\right)(s)\right]^{p}}\int_{-\infty}^{+\infty}d\varphi\,|\varphi|^p\frac{\exp\left[-\frac{\varphi^2}{4}\right]}{\sqrt{4 \pi}}\\
 &=\frac{2^{1+p}}{\sqrt\pi}\Gamma\left(\frac{p+1}{2}\right)\int_0^1d s\frac{s^\frac{p}{2}(1-s)^\frac{p}{2}}{\left[\left(\rho\circ\Phi_\rho^{-1}\right)(s)\right]^{p}}
 \label{costoas}\\
 &=\frac{2^{1+p}}{\sqrt\pi}\Gamma\left(\frac{p+1}{2}\right)\int_0^1 d x \frac{\Phi_\rho^\frac{p}{2}(x)(1-\Phi_\rho(x))^\frac{p}{2}}{\rho^{p-1}(x)}.
\end{split}
\end{equation}
%that appears as a generalization of the expression obtained, using the approach sketched above, in Ref.~\cite{Caracciolo2015}.
In the particular case of a flat distribution the average cost converges to
\begin{equation}
\overline{E_p}  = \frac{2^{1+p}}{\sqrt{\pi}}\,  \Gamma\left(\frac{p+1}{2}\right) \, 
\int_0^1 ds\, \left[s (1-s)\right]^{\frac{p}{2}}  = 2\,  \frac{\Gamma\left(\frac{p}{2}+1\right) }{p+1}
\end{equation}
which is two times the value of the optimal matching. For $p=2$ this gives $\overline{E_2}=2/3$, according to exact result~\reff{exact}. 
Formula~\reff{varphi-d} becomes
\begin{equation}
p_s(x) = \overline{\delta(\varphi(s)-x)} = \frac{e^{-\frac{x^2}{4s(1-s)}}}{\sqrt{4 \pi s(1-s)}}
%\overline{\; \cdot \;} \, \equiv \int_{-\infty}^{+\infty} dx \, p_{B(s)}(x) \;\, \cdot 
\end{equation}
and similarly, see for example~\cite[Appendix A]{Caracciolo:160}, it can be derived that the joint probability distribution $p_{t,s}(x,y)$ for $\varphi(s)$ is (for $t<s$) a bivariate Gaussian distribution %(for $s<t$)
\begin{align}
\begin{split}
p_{t,s}(x,y) & = \overline{\delta(\varphi(t)-x) \, \delta(\varphi(s)-y)} \\
& = \frac{e^{-\frac{x^2}{4t}-\frac{(x-y)^2}{4(s-t)} - \frac{y^2}{4(1-s)} } }{4 \pi \sqrt{t(s-t)(1-s)}}.
\end{split}
\end{align}
%with probability density function equal to
%\begin{equation}
%p_{B(s)}(x) = \frac{e^{-\frac{x^2}{2s(1-s)}}}{\sqrt{2 \pi s(1-s)}} \,.
%\end{equation}
%As we have anticipated, the additional 2 factor in (\ref{Ep}) is due to the fact that the two matchings give equal contributions. 
This allows to compute, for a generic $p>1$, the average of the square of the re-scaled optimal cost 
\be
\label{eq:costvariance}
\overline{E^2_p}=
4 \int_0^1 \!dt \int_0^1 ds \, \overline{\left|\varphi(s)\right|^p \left|\varphi(t)\right|^p},
%4 \left[\int_0^1 ds \overline{\left|\phi(s)\right|^p}\right]^2
\ee
which is 4 times the corresponding one of a bipartite matching problem. 
 %\begin{equation}
%\begin{split}
%&p(x,y)\,dx\,dy \equiv \mathrm{Pr}[\phi(s) \in (x,x+dx),\, \phi(t) \in (y,y+dy)]\\
%\frac{1}{4 \pi  \sqrt{\text{s}}\sqrt{1-\text{t}} \sqrt{\text{t}-\text{s}}}
%\exp  \left(\frac{\frac{\text{t}x^2}{\text{s}}+\frac{(\text{s}-1) y^2}{\text{t}-1}-2 x y}{4(\text{s}-\text{t})}\right)
%&=\frac{1}{4 \pi  \sqrt{s(1-t)(t-s)}}
%\exp  \left(
%-\frac{t}{4s(t-s)}x^2-\frac{1-s}{4(1-t)(t-s)}y^2+\frac{1}{2(t-s)}xy
%-\frac{t x^2/s + (1-s)/(1-t) y^2 - 2 xy}{4(t-s)}
%\right)dx\,dy.
%\end{split}
%\end{equation}
%
In the case $p=2$, the average in Eq.~(\ref{eq:costvariance}) can be evaluated by  using the Wick theorem for expectation values in a Gaussian distribution
\be
\overline{E^2_2}=
4 \int_0^1 \!ds \int_0^s \!dt 
\int_{-\infty}^{\infty}\!dx\,dy\,p_{t,s}(x,y) \, x^2 y^2=
\frac{4}{5} \,,
\ee
and therefore
\be
\overline{E^2_2}-\overline{E_2}^2 = \frac{16}{45} = 0.3\bar{5}.
\ee
This result is in agreement with the numerical simulations (see inset of Fig.~\ref{plot}) and proves that the re-scaled optimal cost is not a self-averaging quantity.

\section{Conclusion and perspectives}\label{sec:conclusions}
In this work we studied the random Euclidean bipartite TSP in one dimension using a weight function which is a power $p$ of the Euclidean distance between red and blue points.
The complete bipartite graph is a special case of a more general problem.
The motivation of this choice is double: on one hand in the one dimensional case we have been able to address clearly the connection between this problem and the assignment and on the other hand we expect the bipartite TSP to be more easily tractable than its monopartite counterpart in more than one dimension. Travelling salesman problems on bipartite graphs may also turn out in practical situations (for instance, a vehicle needing to visit a set of destinations and a set of charging stations).
We provide an explicit solution in the convex case $p>1$, giving the best cycle for each disorder instance of the problem. This allowed us to compute explicitly the average optimal cost when $p=2$ and for every number of points $N$. Interestingly, the value of average optimal cost turned out to be twice the average optimal cost of the assignment problem. In the continuum limit we were also able to find the average optimal cost for generic exponent $p$, using the relation of the one-dimensional assignment with the Brownian bridge process \cite{Caracciolo:159}. In the same thermodynamic limit we computed the variance of the distribution of the optimal costs; since we get a non-vanishing result, we deduce that the average optimal cost is not a self-averaging quantity. This feature is present also in the case of the assignment problem, where the average optimal cost has been shown to be self-averaging only in $d>2$~\cite{Houdayer1998}.
%This work confirms our expectations: the link between the matching problem and the TSP seems to exist only in the bipartite versions of the problems. The idea that problems on bipartite graphs are more tractable  can be useful not only in the field of combinatorial optimization problems, but also in the field of networks and complex systems, in which the need for finding a perfect matching or an optimal Hamiltonian path on a graph is quite common. 

In the field of combinatorial optimization problems, especially in mean field cases (i.e., where the random variables are not correlated), the theory of spin glasses and disordered systems can be used to calculate statistical properties of the optimal solution analytically \cite{mezard1987spin}. In such cases this approach also sheds light on the design of new algorithms to find solutions \cite{mezard2009information}. However, it is not clear in general how to apply these techniques (beyond expanding around the mean field case \cite{mezard1988euclidean,lucibello2017one}), when correlations play an important role, as happens when the graph is embedded in Euclidean spaces. For other problems besides the TSP, analysis of the one-dimensional case has enabled progress in the study of higher-dimensional cases \cite{Caracciolo:163}. As a consequence, a relevant question is whether the relations we obtained in one dimension continue to exist also in $d > 1$, where the bipartite TSP is an NP-complete problem. Recently, we computed exactly the cost and a two-point correlation function in $d=2$ for the assignment problem~\cite{Caracciolo:158, Caracciolo:162, Caracciolo:163}. The investigation of the connections between these two combinatorial optimization problems is material for future work.

%In the field of combinatorial optimization problems, specially in mean field cases, where the random variables of the problem are not correlated, the general theory of spin glasses and disordered systems could be used to calculate statistical properties of the optimal solution at the analytical level \cite{mezard1987spin}, but also to shed light on the design of new algorithms to find their solution \cite{mezard2009information}. However, apart for some attempt to expand around the mean field case \cite{mezard1988euclidean,lucibello2017one}, it is not clear how to apply these techniques when correlations play an important role, as in graph embedded in Euclidean spaces. 
%Our results could be relevant from this point of view, since the insight obtained analyzing the one-dimensional version of an ``Euclidean'' problem has proved to be very useful to tackle the same problem in more than one dimension, when a general procedure is, so far, unknown \cite{Caracciolo:163}.
%As a consequence, an important and interesting question is that if the relations we obtained in one dimension continue to exist also in $d>1$, where the bipartite TSP is really an NP-complete problem. Indeed, we have, recently, got some insight for the assignment problem~\cite{Caracciolo:158, Caracciolo:162, Caracciolo:163} in particular in $d=2$ where the cost and a two-point correlation function have been exactly computed. The investigation of these connections between these two combinatorial optimization problems is left for future work.

\begin{figure}[t]
	\centering
	\includegraphics[width=0.7\columnwidth]{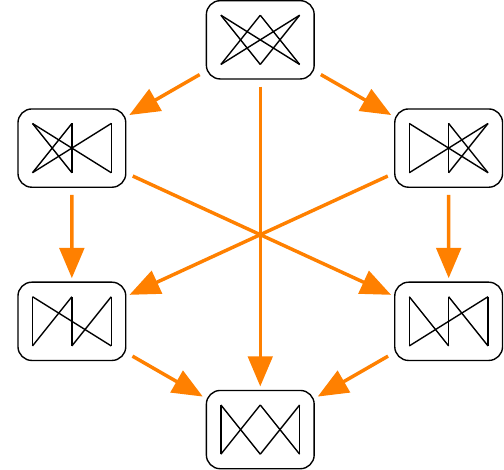}
	\caption{The whole diagram describing the $N=3$ case. In the squared boxes the various cycle configurations are represented. Lower boxes correspond to lower costs. All the possible moves suggested in Lemma~\ref{lemma1} are represented by orange arrows. }\label{fig_poset}
\end{figure}

\section*{Acknowledgments}
The authors are grateful to Riccardo Capelli for useful advices regarding the simulations performed.
\appendix

\section{The case $N=3$}\label{app:n3}

In the case $N=2$ there is only one Hamiltonian cycle, that is $\tilde{h}$. The first nontrivial case is $N=3$. There are 6 Hamiltonian cycles. If we fix the starting point to be $r_1$ there are only two possibilities for the permutation $\sigma$ of the red points, that is $(1,2,3)$ and $(1,3,2)$. One is the dual of the other. We can restrict to the $(1,3,2)$ by removing the degeneracy in the orientation of the cycles. Indeed $\tilde{\sigma}$ is exactly $(1,3,2)$ according to~\reff{sigmatilde}. With this choice the 6 cycles are in correspondence with the permutations $\pi \in \mathcal{S}_3$ of the blue points. We sort in increasing order both the blue and red points. We have
\begin{align}
\begin{split}
E(\pi) & = |r_1-b_{\pi(1)}|^p + |r_1-b_{\pi(3)}|^p +  |r_3-b_{\pi(2)}|^p \\
& + |r_3-b_{\pi(1)}|^p + |r_2-b_{\pi(3)} |^p  + |r_2-b_{\pi(2)}|^p \, .
\end{split}
\end{align}

The optimal solution is $\tilde{\pi} = (2,3,1)$. The permutations $(1,3,2)$ and $(3,2,1)$ have always a grater cost than $\tilde{\pi}$, indeed the corresponding cycles are $(r_1 b_1 \textcolor{orange}{r_3b_3r_2}b_2)$ and $(r_1\textcolor{orange}{b_3r_3b_2}r_2b_1)$, where we have colored in orange the path that, according to Lemma~\ref{lemma1}, can be reversed to lower the total cost. Doing this we obtain the optimal cycle in both cases. 
Notice that, since we can label each cycle using only the $\pi$ permutation, we can restrict ourself to moves that only involve blue points. Since there are three blue points, these moves will always reverse paths of the form $b_ir_jb_k$, so they correspond simply to a swap in the permutation $\pi$. Therefore our moves cannot be used to reach the optimal cycle from every starting cycle. A diagram showing all the possible moves is shown in Fig.~\ref{fig_poset}. In conclusion, the cost function makes $\mathcal{S}_3$ a {\em poset} with an absolute minimum and an absolute maximum.
The permutation $(2,3,1)$ is preceded by both $(1,3,2)$ and $(3,2,1)$, which cannot be compared between them, but both precede $(1,2,3)$ and $(3,1,2)$, which cannot be compared between them. $(2,1,3)$ is the greatest element. 

We compute the average costs for all the permutations. Using the same techniques used in section \ref{sec:costev}, we get that, for the $p=2$ case:
\begin{align}
\begin{split}
\overline{E[(2,3,1)]}  = \frac{3}{4}  & < \overline{E[(1,3,2)]} = \overline{E[(3,2,1)]} = \frac{7}{8} \\
& < \overline{E[(1,2,3)]} = \overline{E[(3,1,2)]} = \frac{9}{8} \\
& < \overline{E[(2,1,3)]} = \frac{5}{4}\, .
\end{split}
\end{align}
\vspace{0.5cm}

\bibliography{AssignmentANDTsp}

\begin{thebibliography}{35}%
\makeatletter
\providecommand \@ifxundefined [1]{%
 \@ifx{#1\undefined}
}%
\providecommand \@ifnum [1]{%
 \ifnum #1\expandafter \@firstoftwo
 \else \expandafter \@secondoftwo
 \fi
}%
\providecommand \@ifx [1]{%
 \ifx #1\expandafter \@firstoftwo
 \else \expandafter \@secondoftwo
 \fi
}%
\providecommand \natexlab [1]{#1}%
\providecommand \enquote  [1]{``#1''}%
\providecommand \bibnamefont  [1]{#1}%
\providecommand \bibfnamefont [1]{#1}%
\providecommand \citenamefont [1]{#1}%
\providecommand \href@noop [0]{\@secondoftwo}%
\providecommand \href [0]{\begingroup \@sanitize@url \@href}%
\providecommand \@href[1]{\@@startlink{#1}\@@href}%
\providecommand \@@href[1]{\endgroup#1\@@endlink}%
\providecommand \@sanitize@url [0]{\catcode `\\12\catcode `\$12\catcode
  `\&12\catcode `\#12\catcode `\^12\catcode `\_12\catcode `\%12\relax}%
\providecommand \@@startlink[1]{}%
\providecommand \@@endlink[0]{}%
\providecommand \url  [0]{\begingroup\@sanitize@url \@url }%
\providecommand \@url [1]{\endgroup\@href {#1}{\urlprefix }}%
\providecommand \urlprefix  [0]{URL }%
\providecommand \Eprint [0]{\href }%
\providecommand \doibase [0]{http://dx.doi.org/}%
\providecommand \selectlanguage [0]{\@gobble}%
\providecommand \bibinfo  [0]{\@secondoftwo}%
\providecommand \bibfield  [0]{\@secondoftwo}%
\providecommand \translation [1]{[#1]}%
\providecommand \BibitemOpen [0]{}%
\providecommand \bibitemStop [0]{}%
\providecommand \bibitemNoStop [0]{.\EOS\space}%
\providecommand \EOS [0]{\spacefactor3000\relax}%
\providecommand \BibitemShut  [1]{\csname bibitem#1\endcsname}%
\let\auto@bib@innerbib\@empty
%</preamble>
\bibitem [{\citenamefont {Lawler}\ \emph {et~al.}(1985)\citenamefont {Lawler},
  \citenamefont {Shmoys}, \citenamefont {Kan},\ and\ \citenamefont
  {Lenstra}}]{lawler1985}%
  \BibitemOpen
  \bibfield  {author} {\bibinfo {author} {\bibfnamefont {E.}~\bibnamefont
  {Lawler}}, \bibinfo {author} {\bibfnamefont {D.}~\bibnamefont {Shmoys}},
  \bibinfo {author} {\bibfnamefont {A.}~\bibnamefont {Kan}}, \ and\ \bibinfo
  {author} {\bibfnamefont {J.}~\bibnamefont {Lenstra}},\ }\href
  {https://books.google.it/books?id=BXBGAAAAYAAJ} {\emph {\bibinfo {title} {The
  Traveling Salesman Problem}}}\ (\bibinfo  {publisher} {John Wiley \& Sons,
  Incorporated},\ \bibinfo {year} {1985})\BibitemShut {NoStop}%
\bibitem [{\citenamefont {Menger}(1932)}]{menger1932}%
  \BibitemOpen
  \bibfield  {author} {\bibinfo {author} {\bibfnamefont {K.}~\bibnamefont
  {Menger}},\ }\href@noop {} {\bibfield  {journal} {\bibinfo  {journal}
  {Ergebnisse eines mathematischen kolloquiums}\ }\textbf {\bibinfo {volume}
  {2}},\ \bibinfo {pages} {11} (\bibinfo {year} {1932})}\BibitemShut {NoStop}%
\bibitem [{Note1()}]{Note1}%
  \BibitemOpen
  \bibinfo {note}
  {Http://www.claymath.org/millennium-problems/p-vs-np-problem}\BibitemShut
  {NoStop}%
\bibitem [{\citenamefont {Reinelt}(1994)}]{reinelt1994}%
  \BibitemOpen
  \bibfield  {author} {\bibinfo {author} {\bibfnamefont {G.}~\bibnamefont
  {Reinelt}},\ }\href@noop {} {\emph {\bibinfo {title} {The traveling salesman:
  computational solutions for TSP applications}}}\ (\bibinfo  {publisher}
  {Springer-Verlag},\ \bibinfo {year} {1994})\BibitemShut {NoStop}%
\bibitem [{\citenamefont {Papadimitriou}(1977)}]{papadimitriou1977}%
  \BibitemOpen
  \bibfield  {author} {\bibinfo {author} {\bibfnamefont {C.~H.}\ \bibnamefont
  {Papadimitriou}},\ }\href@noop {} {\bibfield  {journal} {\bibinfo  {journal}
  {Theoretical computer science}\ }\textbf {\bibinfo {volume} {4}},\ \bibinfo
  {pages} {237} (\bibinfo {year} {1977})}\BibitemShut {NoStop}%
\bibitem [{\citenamefont {Kirkpatrick}\ \emph {et~al.}(1983)\citenamefont
  {Kirkpatrick}, \citenamefont {Gelatt},\ and\ \citenamefont
  {Vecchi}}]{Kirkpatrick1983}%
  \BibitemOpen
  \bibfield  {author} {\bibinfo {author} {\bibfnamefont {S.}~\bibnamefont
  {Kirkpatrick}}, \bibinfo {author} {\bibfnamefont {C.}~\bibnamefont {Gelatt}},
  \ and\ \bibinfo {author} {\bibfnamefont {M.}~\bibnamefont {Vecchi}},\ }\href
  {http://www.jstor.org/stable/1690046} {\bibfield  {journal} {\bibinfo
  {journal} {Science}\ }\textbf {\bibinfo {volume} {220}},\ \bibinfo {pages}
  {671} (\bibinfo {year} {1983})}\BibitemShut {NoStop}%
\bibitem [{\citenamefont {Sourlas}(1986)}]{Sourlas1986}%
  \BibitemOpen
  \bibfield  {author} {\bibinfo {author} {\bibfnamefont {N.}~\bibnamefont
  {Sourlas}},\ }\href@noop {} {\bibfield  {journal} {\bibinfo  {journal}
  {Europhysics Letters}\ }\textbf {\bibinfo {volume} {2}},\ \bibinfo {pages}
  {919} (\bibinfo {year} {1986})}\BibitemShut {NoStop}%
\bibitem [{\citenamefont {Mézard}\ \emph {et~al.}(1987)\citenamefont
  {Mézard}, \citenamefont {Parisi},\ and\ \citenamefont
  {Virasoro}}]{mezard1987spin}%
  \BibitemOpen
  \bibfield  {author} {\bibinfo {author} {\bibfnamefont {M.}~\bibnamefont
  {Mézard}}, \bibinfo {author} {\bibfnamefont {G.}~\bibnamefont {Parisi}}, \
  and\ \bibinfo {author} {\bibfnamefont {M.}~\bibnamefont {Virasoro}},\
  }\href@noop {} {\emph {\bibinfo {title} {Spin glass theory and beyond: An
  Introduction to the Replica Method and Its Applications}}},\ Vol.~\bibinfo
  {volume} {9}\ (\bibinfo  {publisher} {World Scientific Publishing Company},\
  \bibinfo {year} {1987})\BibitemShut {NoStop}%
\bibitem [{\citenamefont {Mezard}\ and\ \citenamefont
  {Montanari}(2009)}]{mezard2009information}%
  \BibitemOpen
  \bibfield  {author} {\bibinfo {author} {\bibfnamefont {M.}~\bibnamefont
  {Mezard}}\ and\ \bibinfo {author} {\bibfnamefont {A.}~\bibnamefont
  {Montanari}},\ }\href@noop {} {\emph {\bibinfo {title} {Information, physics,
  and computation}}}\ (\bibinfo  {publisher} {Oxford University Press},\
  \bibinfo {year} {2009})\BibitemShut {NoStop}%
\bibitem [{\citenamefont {Vannimenus}\ and\ \citenamefont
  {M{\'{e}}zard}(1984)}]{Vannimenus1984}%
  \BibitemOpen
  \bibfield  {author} {\bibinfo {author} {\bibfnamefont {J.}~\bibnamefont
  {Vannimenus}}\ and\ \bibinfo {author} {\bibfnamefont {M.}~\bibnamefont
  {M{\'{e}}zard}},\ }\href@noop {} {\bibfield  {journal} {\bibinfo  {journal}
  {Journal de Physique Lettres}\ }\textbf {\bibinfo {volume} {45}},\ \bibinfo
  {pages} {L1145} (\bibinfo {year} {1984})}\BibitemShut {NoStop}%
\bibitem [{\citenamefont {Orland}(1985)}]{Orland1985}%
  \BibitemOpen
  \bibfield  {author} {\bibinfo {author} {\bibfnamefont {H.}~\bibnamefont
  {Orland}},\ }\href@noop {} {\bibfield  {journal} {\bibinfo  {journal} {Le
  Journal de Physique - Lettres}\ }\textbf {\bibinfo {volume} {46}} (\bibinfo
  {year} {1985})}\BibitemShut {NoStop}%
\bibitem [{\citenamefont {M{\'{e}}zard}\ and\ \citenamefont
  {Parisi}(1986{\natexlab{a}})}]{Mezard1986}%
  \BibitemOpen
  \bibfield  {author} {\bibinfo {author} {\bibfnamefont {M.}~\bibnamefont
  {M{\'{e}}zard}}\ and\ \bibinfo {author} {\bibfnamefont {G.}~\bibnamefont
  {Parisi}},\ }\href
  {http://jphys.journaldephysique.org/index.php?option=com{\_}article{\&}access=doi{\&}doi=10.1051/jphys:019860047080128500}
  {\bibfield  {journal} {\bibinfo  {journal} {Journal de Physique}\ }\textbf
  {\bibinfo {volume} {47}},\ \bibinfo {pages} {1285} (\bibinfo {year}
  {1986}{\natexlab{a}})}\BibitemShut {NoStop}%
\bibitem [{\citenamefont {M{\'{e}}zard}\ and\ \citenamefont
  {Parisi}(1986{\natexlab{b}})}]{Mezard1986a}%
  \BibitemOpen
  \bibfield  {author} {\bibinfo {author} {\bibfnamefont {M.}~\bibnamefont
  {M{\'{e}}zard}}\ and\ \bibinfo {author} {\bibfnamefont {G.}~\bibnamefont
  {Parisi}},\ }\href@noop {} {\bibfield  {journal} {\bibinfo  {journal}
  {Europhysics Letters}\ }\textbf {\bibinfo {volume} {2}},\ \bibinfo {pages}
  {913} (\bibinfo {year} {1986}{\natexlab{b}})}\BibitemShut {NoStop}%
\bibitem [{\citenamefont {Krauth}\ and\ \citenamefont
  {M{\'{e}}zard}(1989)}]{Krauth1989}%
  \BibitemOpen
  \bibfield  {author} {\bibinfo {author} {\bibfnamefont {W.}~\bibnamefont
  {Krauth}}\ and\ \bibinfo {author} {\bibfnamefont {M.}~\bibnamefont
  {M{\'{e}}zard}},\ }\href
  {http://stacks.iop.org/0295-5075/8/i=3/a=002?key=crossref.5926a172ce9c8b95f33b75c6996541e7}
  {\bibfield  {journal} {\bibinfo  {journal} {Europhysics Letters}\ }\textbf
  {\bibinfo {volume} {8}},\ \bibinfo {pages} {213} (\bibinfo {year}
  {1989})}\BibitemShut {NoStop}%
\bibitem [{\citenamefont {Ravanbakhsh}\ \emph {et~al.}(2014)\citenamefont
  {Ravanbakhsh}, \citenamefont {Rabbany},\ and\ \citenamefont
  {Greiner}}]{Ravanbakhsh2014}%
  \BibitemOpen
  \bibfield  {author} {\bibinfo {author} {\bibfnamefont {S.}~\bibnamefont
  {Ravanbakhsh}}, \bibinfo {author} {\bibfnamefont {R.}~\bibnamefont
  {Rabbany}}, \ and\ \bibinfo {author} {\bibfnamefont {R.}~\bibnamefont
  {Greiner}},\ }\href {http://arxiv.org/abs/1406.0941} {\bibfield  {journal}
  {\bibinfo  {journal} {Advances in Neural Information Processing Systems}\
  }\textbf {\bibinfo {volume} {1}},\ \bibinfo {pages} {289} (\bibinfo {year}
  {2014})},\ \Eprint {http://arxiv.org/abs/1406.0941} {arXiv:1406.0941}
  \BibitemShut {NoStop}%
\bibitem [{\citenamefont {Wastlund}(2010)}]{Wastlund}%
  \BibitemOpen
  \bibfield  {author} {\bibinfo {author} {\bibfnamefont {J.}~\bibnamefont
  {Wastlund}},\ }\href {\doibase {10.1007/s11511-010-0046-7}} {\bibfield
  {journal} {\bibinfo  {journal} {Acta Mathematica}\ }\textbf {\bibinfo
  {volume} {204}},\ \bibinfo {pages} {91} (\bibinfo {year}
  {{2010}})}\BibitemShut {NoStop}%
\bibitem [{\citenamefont {Beardwood}\ \emph {et~al.}(1959)\citenamefont
  {Beardwood}, \citenamefont {Halton},\ and\ \citenamefont {Hammersley}}]{BHH}%
  \BibitemOpen
  \bibfield  {author} {\bibinfo {author} {\bibfnamefont {J.}~\bibnamefont
  {Beardwood}}, \bibinfo {author} {\bibfnamefont {J.~H.}\ \bibnamefont
  {Halton}}, \ and\ \bibinfo {author} {\bibfnamefont {J.~M.}\ \bibnamefont
  {Hammersley}},\ }\href@noop {} {\bibfield  {journal} {\bibinfo  {journal}
  {Proc. Cambridge Philos. Soc.}\ }\textbf {\bibinfo {volume} {55}} (\bibinfo
  {year} {1959})}\BibitemShut {NoStop}%
\bibitem [{\citenamefont {Steele}(1981)}]{Steele}%
  \BibitemOpen
  \bibfield  {author} {\bibinfo {author} {\bibfnamefont {M.}~\bibnamefont
  {Steele}},\ }\href@noop {} {\bibfield  {journal} {\bibinfo  {journal} {Ann.
  Probability}\ }\textbf {\bibinfo {volume} {9}},\ \bibinfo {pages} {365}
  (\bibinfo {year} {1981})}\BibitemShut {NoStop}%
\bibitem [{\citenamefont {Karp}\ and\ \citenamefont {Steele}(1985)}]{Karp}%
  \BibitemOpen
  \bibfield  {author} {\bibinfo {author} {\bibfnamefont {R.~M.}\ \bibnamefont
  {Karp}}\ and\ \bibinfo {author} {\bibfnamefont {M.}~\bibnamefont {Steele}},\
  }\enquote {\bibinfo {title} {The travelling salesman problem},}\ \ (\bibinfo
  {publisher} {John Wiley and Sons, New York},\ \bibinfo {year}
  {1985})\BibitemShut {NoStop}%
\bibitem [{\citenamefont {Percus}\ and\ \citenamefont
  {Martin}(1996)}]{Percus1996}%
  \BibitemOpen
  \bibfield  {author} {\bibinfo {author} {\bibfnamefont {A.~G.}\ \bibnamefont
  {Percus}}\ and\ \bibinfo {author} {\bibfnamefont {O.~C.}\ \bibnamefont
  {Martin}},\ }\href {http://www.ncbi.nlm.nih.gov/pubmed/10061658} {\bibfield
  {journal} {\bibinfo  {journal} {Physical Review Letters}\ }\textbf {\bibinfo
  {volume} {76}},\ \bibinfo {pages} {1188} (\bibinfo {year}
  {1996})}\BibitemShut {NoStop}%
\bibitem [{\citenamefont {Cerf}\ \emph {et~al.}(1997)\citenamefont {Cerf},
  \citenamefont {{Boutet de Monvel}}, \citenamefont {Bohigas}, \citenamefont
  {Martin},\ and\ \citenamefont {Percus}}]{Cerf1997}%
  \BibitemOpen
  \bibfield  {author} {\bibinfo {author} {\bibfnamefont {N.~J.}\ \bibnamefont
  {Cerf}}, \bibinfo {author} {\bibfnamefont {J.~H.}\ \bibnamefont {{Boutet de
  Monvel}}}, \bibinfo {author} {\bibfnamefont {O.}~\bibnamefont {Bohigas}},
  \bibinfo {author} {\bibfnamefont {O.~C.}\ \bibnamefont {Martin}}, \ and\
  \bibinfo {author} {\bibfnamefont {A.~G.}\ \bibnamefont {Percus}},\ }\href
  {http://jp1.journaldephysique.org/articles/jp1/abs/1997/01/jp1v7p117/jp1v7p117.html}
  {\bibfield  {journal} {\bibinfo  {journal} {Journal de Physique I}\ }\textbf
  {\bibinfo {volume} {7}},\ \bibinfo {pages} {117} (\bibinfo {year}
  {1997})}\BibitemShut {NoStop}%
\bibitem [{\citenamefont {Caracciolo}\ \emph {et~al.}(2014)\citenamefont
  {Caracciolo}, \citenamefont {Lucibello}, \citenamefont {Parisi},\ and\
  \citenamefont {Sicuro}}]{Caracciolo:158}%
  \BibitemOpen
  \bibfield  {author} {\bibinfo {author} {\bibfnamefont {S.}~\bibnamefont
  {Caracciolo}}, \bibinfo {author} {\bibfnamefont {C.}~\bibnamefont
  {Lucibello}}, \bibinfo {author} {\bibfnamefont {G.}~\bibnamefont {Parisi}}, \
  and\ \bibinfo {author} {\bibfnamefont {G.}~\bibnamefont {Sicuro}},\ }\href
  {\doibase 10.1103/PhysRevE.90.012118} {\bibfield  {journal} {\bibinfo
  {journal} {Phys. Rev. E}\ }\textbf {\bibinfo {volume} {90}},\ \bibinfo
  {pages} {012118} (\bibinfo {year} {2014})},\ \Eprint
  {http://arxiv.org/abs/arXiv:1402.6993} {arXiv:1402.6993} \BibitemShut
  {NoStop}%
\bibitem [{\citenamefont {Boniolo}\ \emph {et~al.}(2014)\citenamefont
  {Boniolo}, \citenamefont {Caracciolo},\ and\ \citenamefont
  {Sportiello}}]{Caracciolo:159}%
  \BibitemOpen
  \bibfield  {author} {\bibinfo {author} {\bibfnamefont {E.}~\bibnamefont
  {Boniolo}}, \bibinfo {author} {\bibfnamefont {S.}~\bibnamefont {Caracciolo}},
  \ and\ \bibinfo {author} {\bibfnamefont {A.}~\bibnamefont {Sportiello}},\
  }\href {\doibase 10.1088/1742-5468/2014/11/P11023} {\bibfield  {journal}
  {\bibinfo  {journal} {J. Stat. Mech.}\ }\textbf {\bibinfo {volume} {11}},\
  \bibinfo {pages} {P11023} (\bibinfo {year} {2014})},\ \Eprint
  {http://arxiv.org/abs/arXiv:1403.1836} {arXiv:1403.1836} \BibitemShut
  {NoStop}%
\bibitem [{\citenamefont {Halton}(1995)}]{halton1995}%
  \BibitemOpen
  \bibfield  {author} {\bibinfo {author} {\bibfnamefont {J.~H.}\ \bibnamefont
  {Halton}},\ }\href {\doibase 10.1007/BF03024786} {\bibfield  {journal}
  {\bibinfo  {journal} {The Mathematical Intelligencer}\ }\textbf {\bibinfo
  {volume} {17}},\ \bibinfo {pages} {36} (\bibinfo {year} {1995})}\BibitemShut
  {NoStop}%
\bibitem [{\citenamefont {Misiurewicz}(1996)}]{misiurewicz1996}%
  \BibitemOpen
  \bibfield  {author} {\bibinfo {author} {\bibfnamefont {M.}~\bibnamefont
  {Misiurewicz}},\ }\href@noop {} {\bibfield  {journal} {\bibinfo  {journal}
  {"The Mathematical Intelligencer"}\ }\textbf {\bibinfo {volume} {18}},\
  \bibinfo {pages} {32} (\bibinfo {year} {1996})}\BibitemShut {NoStop}%
\bibitem [{\citenamefont {Polster}(2002)}]{polster2002}%
  \BibitemOpen
  \bibfield  {author} {\bibinfo {author} {\bibfnamefont {B.}~\bibnamefont
  {Polster}},\ }\href@noop {} {\bibfield  {journal} {\bibinfo  {journal}
  {Nature}\ }\textbf {\bibinfo {volume} {420}},\ \bibinfo {pages} {476}
  (\bibinfo {year} {2002})}\BibitemShut {NoStop}%
\bibitem [{\citenamefont {García}\ and\ \citenamefont
  {Tejel}(2017)}]{garcia2017}%
  \BibitemOpen
  \bibfield  {author} {\bibinfo {author} {\bibfnamefont {A.}~\bibnamefont
  {García}}\ and\ \bibinfo {author} {\bibfnamefont {J.}~\bibnamefont
  {Tejel}},\ }\href {\doibase 10.1016/j.ejor.2016.07.060} {\bibfield  {journal}
  {\bibinfo  {journal} {European Journal of Operational Research}\ }\textbf
  {\bibinfo {volume} {257}},\ \bibinfo {pages} {429} (\bibinfo {year}
  {2017})}\BibitemShut {NoStop}%
\bibitem [{\citenamefont {{McCann Robert}}(1999)}]{McCannRobert1999}%
  \BibitemOpen
  \bibfield  {author} {\bibinfo {author} {\bibnamefont {{McCann Robert}}},\
  }\href@noop {} {\bibfield  {journal} {\bibinfo  {journal} {Proc.R.Soc.A:
  Math.,Phys. Eng Sci}\ }\textbf {\bibinfo {volume} {455}},\ \bibinfo {pages}
  {1341} (\bibinfo {year} {1999})}\BibitemShut {NoStop}%
\bibitem [{\citenamefont {Caracciolo}\ and\ \citenamefont
  {Sicuro}(2014)}]{Caracciolo:160}%
  \BibitemOpen
  \bibfield  {author} {\bibinfo {author} {\bibfnamefont {S.}~\bibnamefont
  {Caracciolo}}\ and\ \bibinfo {author} {\bibfnamefont {G.}~\bibnamefont
  {Sicuro}},\ }\href {\doibase 10.1103/PhysRevE.90.042112} {\bibfield
  {journal} {\bibinfo  {journal} {Phys. Rev. E}\ }\textbf {\bibinfo {volume}
  {90}} (\bibinfo {year} {2014}),\ 10.1103/PhysRevE.90.042112},\ \Eprint
  {http://arxiv.org/abs/arXiv:1406.7565} {arXiv:1406.7565} \BibitemShut
  {NoStop}%
\bibitem [{\citenamefont {Caracciolo}\ \emph {et~al.}(2017)\citenamefont
  {Caracciolo}, \citenamefont {D'Achille},\ and\ \citenamefont
  {Sicuro}}]{Caracciolo:169}%
  \BibitemOpen
  \bibfield  {author} {\bibinfo {author} {\bibfnamefont {S.}~\bibnamefont
  {Caracciolo}}, \bibinfo {author} {\bibfnamefont {M.~P.}\ \bibnamefont
  {D'Achille}}, \ and\ \bibinfo {author} {\bibfnamefont {G.}~\bibnamefont
  {Sicuro}},\ }\href {\doibase 10.1103/PhysRevE.96.042102} {\bibfield
  {journal} {\bibinfo  {journal} {Phys. Rev. E}\ }\textbf {\bibinfo {volume}
  {96}},\ \bibinfo {pages} {042102} (\bibinfo {year} {2017})},\ \Eprint
  {http://arxiv.org/abs/arXiv:1707.05541v1} {arXiv:1707.05541v1} \BibitemShut
  {NoStop}%
\bibitem [{\citenamefont {Houdayer}\ \emph {et~al.}(1998)\citenamefont
  {Houdayer}, \citenamefont {{Boutet de Monvel}},\ and\ \citenamefont
  {Martin}}]{Houdayer1998}%
  \BibitemOpen
  \bibfield  {author} {\bibinfo {author} {\bibfnamefont {J.}~\bibnamefont
  {Houdayer}}, \bibinfo {author} {\bibfnamefont {J.}~\bibnamefont {{Boutet de
  Monvel}}}, \ and\ \bibinfo {author} {\bibfnamefont {O.}~\bibnamefont
  {Martin}},\ }\href {\doibase 10.1007/s100510050565} {\bibfield  {journal}
  {\bibinfo  {journal} {The European Physical Journal B}\ }\textbf {\bibinfo
  {volume} {6}},\ \bibinfo {pages} {383} (\bibinfo {year} {1998})},\ \Eprint
  {http://arxiv.org/abs/9803195} {arXiv:9803195 [cond-mat]} \BibitemShut
  {NoStop}%
\bibitem [{\citenamefont {Mézard}\ and\ \citenamefont
  {Parisi}(1988)}]{mezard1988euclidean}%
  \BibitemOpen
  \bibfield  {author} {\bibinfo {author} {\bibfnamefont {M.}~\bibnamefont
  {Mézard}}\ and\ \bibinfo {author} {\bibfnamefont {G.}~\bibnamefont
  {Parisi}},\ }\href@noop {} {\bibfield  {journal} {\bibinfo  {journal}
  {Journal de Physique}\ }\textbf {\bibinfo {volume} {49}},\ \bibinfo {pages}
  {2019} (\bibinfo {year} {1988})}\BibitemShut {NoStop}%
\bibitem [{\citenamefont {Lucibello}\ \emph {et~al.}(2017)\citenamefont
  {Lucibello}, \citenamefont {Parisi},\ and\ \citenamefont
  {Sicuro}}]{lucibello2017one}%
  \BibitemOpen
  \bibfield  {author} {\bibinfo {author} {\bibfnamefont {C.}~\bibnamefont
  {Lucibello}}, \bibinfo {author} {\bibfnamefont {G.}~\bibnamefont {Parisi}}, \
  and\ \bibinfo {author} {\bibfnamefont {G.}~\bibnamefont {Sicuro}},\
  }\href@noop {} {\bibfield  {journal} {\bibinfo  {journal} {Physical Review
  E}\ }\textbf {\bibinfo {volume} {95}},\ \bibinfo {pages} {012302} (\bibinfo
  {year} {2017})}\BibitemShut {NoStop}%
\bibitem [{\citenamefont {Caracciolo}\ and\ \citenamefont
  {Sicuro}(2015{\natexlab{a}})}]{Caracciolo:163}%
  \BibitemOpen
  \bibfield  {author} {\bibinfo {author} {\bibfnamefont {S.}~\bibnamefont
  {Caracciolo}}\ and\ \bibinfo {author} {\bibfnamefont {G.}~\bibnamefont
  {Sicuro}},\ }\href {\doibase
  http://dx.doi.org/10.1103/PhysRevLett.115.230601} {\bibfield  {journal}
  {\bibinfo  {journal} {Phys. Rev. Lett.}\ }\textbf {\bibinfo {volume} {115}},\
  \bibinfo {pages} {230601} (\bibinfo {year} {2015}{\natexlab{a}})},\ \Eprint
  {http://arxiv.org/abs/arXiv:1510.02320} {arXiv:1510.02320} \BibitemShut
  {NoStop}%
\bibitem [{\citenamefont {Caracciolo}\ and\ \citenamefont
  {Sicuro}(2015{\natexlab{b}})}]{Caracciolo:162}%
  \BibitemOpen
  \bibfield  {author} {\bibinfo {author} {\bibfnamefont {S.}~\bibnamefont
  {Caracciolo}}\ and\ \bibinfo {author} {\bibfnamefont {G.}~\bibnamefont
  {Sicuro}},\ }\href {\doibase http://dx.doi.org/10.1103/PhysRevE.91.062125}
  {\bibfield  {journal} {\bibinfo  {journal} {Phys. Rev. E}\ }\textbf {\bibinfo
  {volume} {91}} (\bibinfo {year} {2015}{\natexlab{b}}),\
  http://dx.doi.org/10.1103/PhysRevE.91.062125},\ \Eprint
  {http://arxiv.org/abs/arXiv:1504.00614} {arXiv:1504.00614} \BibitemShut
  {NoStop}%
\end{thebibliography}%

\end{document}